\documentclass[12pt,a4paper]{amsart}

\usepackage[hmargin=2cm,vmargin=2cm]{geometry}

\usepackage{amsmath,amsfonts,amssymb}

\usepackage[driverfallback=hypertex]{hyperref}
\usepackage{nameref,zref-xr}                    

\newcommand{\mbC}{\mathbb C}

\newcommand{\oM}{\overline{\mathcal M}}
\newcommand{\eps}{\varepsilon}
\def\d{\partial}
\def\CP1{\mathbb{C}\mathrm{P}^1}

\newcommand{\tu}{\widetilde u}

\newcommand{\<}{\left <}
\renewcommand{\>}{\right >}

\newcommand{\of}{\overline f}
\newcommand{\cA}{{\mathcal A}}
\newcommand{\hcA}{{\widehat{\mathcal A}}}
\newcommand{\oh}{{\overline h}}
\newcommand{\hLambda}{{\widehat\Lambda}}
\newcommand{\im}{\mathop{\mathrm{im}}\nolimits}
\newcommand{\ch}{\mathop{\mathrm{ch}}\nolimits}
\newcommand{\sgn}{\mathop{\mathrm{sgn}}\nolimits}
\renewcommand{\deg}{\mathop{\mathrm{deg}}\nolimits}

\newcommand{\tF}{\widetilde F}

\newtheorem{theorem}{Theorem}[section]
\newtheorem{proposition}[theorem]{Proposition}
\newtheorem{lemma}[theorem]{Lemma}

\theoremstyle{definition}

\newtheorem{example}[theorem]{Example}
\newtheorem{remark}[theorem]{Remark}

\numberwithin{equation}{section}

\title{Dubrovin-Zhang hierarchy for the Hodge integrals}

\author{A. Buryak}

\address{A.~Buryak:\newline
Department of Mathematics,
ETH Zurich,\newline
Ramistrasse 101 8092, HG G27.1, Zurich, Switzerland, and\newline
Department of Mathematics, Moscow State University,\newline
Leninskie gory GSP-2 119992, Moscow, Russia} 
\email{buryaksh@gmail.com}

\subjclass[2010]{37K05, 14H10} 

\begin{document}

\begin{abstract}
In this paper we prove that the generating series of the Hodge integrals over the moduli space of stable curves is a solution of a certain deformation of the KdV hierarchy. This hierarchy is constructed in the framework of the Dubrovin-Zhang theory of the hierarchies of the topological type. It occurs that our deformation of the KdV hierarchy is closely related to the hierarchy of the Intermediate Long Wave equation.
\end{abstract}

\maketitle

\section{Introduction}

Let $\oM_{g,n}$ be the moduli space of stable complex algebraic curves with $n$ labelled marked points. The intersection theory of $\oM_{g,n}$ is closely related to the theory of integrable systems of partial differential equations. The basic result in this subject is the famous Witten conjecture (\cite{Wit91}) proved by M.~Kontsevich (see \cite{Kon92}). It tells the following. The class $\psi_i\in H^2(\oM_{g,n};\mbC)$ is defined as the first Chern class of the line bundle over $\oM_{g,n}$ formed by the cotangent lines at the $i$-th marked point. Intersection numbers $\<\tau_{k_1}\tau_{k_2}\ldots\tau_{k_n}\>_g$ are defined as follows:
$$
\<\tau_{k_1}\tau_{k_2}\ldots\tau_{k_n}\>_g:=\int_{\oM_{g,n}}\psi_1^{k_1}\psi_2^{k_2}\ldots\psi_n^{k_n}.
$$ 
Let us introduce variables $\hbar,t_0,t_1,t_2,\ldots$ and consider the generating series 
$$
F:=\sum_{\substack{g\ge 0,n\ge 1\\2g-2+n>0}}\frac{\hbar^g}{n!}\sum_{k_1,\ldots,k_n\ge 0}\<\tau_{k_1}\ldots\tau_{k_n}\>_gt_{k_1}\ldots t_{k_n}.
$$

Witten's conjecture, proved by M.~Kontsevich, says that the second derivative $\frac{\d^2 F}{\d t_0^2}$ is a solution of the KdV hierarchy. The first two equations of this hierarchy are
\begin{align*}
u_{t_1}&=uu_x+\frac{\hbar}{12}u_{xxx},\\
u_{t_2}&=\frac{1}{2}u^2u_x+\frac{\hbar}{12}(2u_xu_{xx}+uu_{xxx})+\frac{\hbar^2}{240}u_{xxxxx}.
\end{align*}
Here we identify $x$ with $t_0$. 

In this paper we study the Hodge integrals over the moduli space $\oM_{g,n}$: 
$$
\<\lambda_j\tau_{k_1}\ldots\tau_{k_n}\>_g:=\int_{\oM_{g,n}}\lambda_j\psi_1^{k_1}\psi_2^{k_2}\ldots\psi_n^{k_n},
$$
where $\lambda_j\in H^{2j}(\oM_{g,n};\mbC)$ is the $j$-th Chern class of the rank $g$ Hodge vector bundle over~$\oM_{g,n}$ whose fibers over smooth curves are the spaces of holomorphic one-forms. Consider the generating series 
$$
F^{Hodge}:=\sum_{\substack{g,n\ge 0\\2g-2+n>0}}\sum_{0\le j\le g}\frac{\hbar^g\eps^j}{n!}\sum_{k_1,\ldots,k_n\ge 0}\<\lambda_j\tau_{k_1}\ldots\tau_{k_n}\>_gt_{k_1}\ldots t_{k_n}.
$$

The main result of the paper is the following. In Section~\ref{subsection: deformed kdv} we construct a certain hamiltonian deformation of the KdV hierarchy. The first two equations of this hierarchy are
\begin{align}
u_{t_1}&=uu_x+\sum_{g\ge 1}\hbar^g\eps^{g-1}\frac{|B_{2g}|}{(2g)!}u_{2g+1},\label{eq:ILW equation}\\
u_{t_2}&=\frac{1}{2}u^2u_x+\sum_{g\ge 1}\frac{|B_{2g}|}{(2g)!}\hbar^g\frac{\eps^{g-1}}{4}(2(u u_{2g})_x+\d_x^{2g+1}(u^2))+\sum_{g\ge 2}\frac{|B_{2g}|}{(2g)!}\hbar^g\eps^{g-2}(g+1)u_{2g+1}.\notag
\end{align}
Here $B_{2g}$ are Bernoulli numbers: $B_2=\frac{1}{6},B_4=-\frac{1}{30},\ldots$; and we denote by $u_i$ the derivative~$\d_x^i u$. We call this hierarchy the deformed KdV hierarchy. Let 
\begin{gather}\label{eq:transformation}
\widetilde F^{Hodge}:=F^{Hodge}+\sum_{g\ge 1}\frac{(-1)^g}{2^{2g}(2g+1)!}\hbar^{g}\eps^g\frac{\d^{2g}F^{Hodge}}{\d t_0^{2g}}.
\end{gather}
\begin{theorem}\label{theorem: main theorem}
The series $\frac{\d^2\widetilde F^{Hodge}}{\d t_0^2}$ is a solution of the deformed KdV hierarchy. 
\end{theorem}
\noindent We remind the reader that we identify $x$ with $t_0$. 

Let us explain how to compute the series~$F^{Hodge}$ using this theorem. Since~$\oM_{0,3}$ is a point and $\int_{\oM_{1,1}}\lambda_1=\frac{1}{24}$, we have
$$
\left.F^{Hodge}\right|_{t_{\ge 1}=0}=\frac{t_0^3}{6}+\frac{\hbar\eps}{24}t_0.
$$
Therefore,
$$
\left.\frac{\d^2 \widetilde F^{Hodge}}{\d t_0^2}\right|_{t_{\ge 1}=0}=t_0.
$$
Using this equation as an initial condition for the deformed KdV hierarchy, Theorem~\ref{theorem: main theorem} allows to determine the series~$\frac{\d^2 \widetilde F^{Hodge}}{\d t_0^2}$. Note that the transformation~\eqref{eq:transformation} is invertible, one can check that
$$
F^{Hodge}=\widetilde F^{Hodge}+\sum_{g\ge 1}\frac{2^{2g-1}-1}{2^{2g-1}}\frac{|B_{2g}|}{(2g)!}\hbar^g\eps^g\frac{\d^{2g}\widetilde F^{Hodge}}{\d t_0^{2g}}.
$$
Therefore, using~$\frac{\d^2 \widetilde F^{Hodge}}{\d t_0^2}$ we can reconstruct~$\frac{\d^2 F^{Hodge}}{\d t_0^2}$. After that the string equation allows to determine~$F^{Hodge}$. This is the same argument as E.~Witten used in~\cite{Wit91} in order to reconstruct the series~$F$ from the second derivative~$\frac{\d^2 F}{\d t_0^2}$. 

\begin{remark}
In~\cite{Kaz09} M.~Kazarian proved that after a certain change of variables the series~$F^{Hodge}$ becomes a solution of the KP hierarchy. It seems to be interesting to relate his result to ours.
\end{remark}

Equation~\eqref{eq:ILW equation} coincides (after several rescalings) with the Intermediate Long Wave (ILW) equation (see e.g. \cite{SAK79}). We are very grateful to S. Ferapontov and D. Novikov for noticing this fact after the author's talk on the conference in Trieste (Hamiltonian PDEs, Frobenius manifolds and Deligne-Mumford moduli spaces, September 2013). An infinite sequence of conserved quantities of the ILW equation was constructed in \cite{SAK79}. We compare these conserved quantities with the Hamiltonians of our deformed KdV hierarchy in Section~\ref{section: ILW}.  

Our approach is based on the B. Dubrovin and Y. Zhang theory of the integrable hierarchies of the topological type. In \cite{DZ05} B. Dubrovin and Y. Zhang gave a construction of a bihamiltonian hierarchy associated to any conformal semisimple Frobenius manifold. They conjectured that the equations and the hamiltonian structures of this hierarchy are polynomial. In \cite{BPS12a} the authors suggested a more general construction of a hamiltonian hierarchy associated to an arbitrary semisimple cohomological field theory and proved the polynomiality of the equations and of the hamiltonian structure (see also~\cite{BPS12b}). One of the simplest examples of a cohomological field theory is the one formed by the Hodge classes
\begin{gather}\label{Hodge classes}
1+\eps\lambda_1+\eps\lambda_2+\ldots+\eps^g\lambda_g\in H^*(\oM_{g,n};\mbC).
\end{gather}  
The main step in the proof of Theorem~\ref{theorem: main theorem} is the application of the polynomiality theorem from \cite{BPS12a} to the Dubrovin-Zhang hierarchy associated to the cohomological field theory~\eqref{Hodge classes}. We also prove the following theorem.

\begin{theorem}\label{theorem: DZ hierarchy}
Consider the Dubrovin-Zhang hierarchy associated to the cohomological field theory \eqref{Hodge classes}. Then the Miura transformation 
\begin{gather}\label{eq:transformation2}
u\mapsto \widetilde u=u+\sum_{g\ge 1}\frac{(-1)^g}{2^{2g}(2g+1)!}\hbar^{g}\eps^g u_{2g}
\end{gather}
transforms this hierarchy to the deformed KdV hierarchy. 
\end{theorem}
One can see that the variable~$\tu$ is related to the variable~$u$ (eq.~\eqref{eq:transformation2}) in the same way as the series~$\tF^{Hodge}$ is related to the series~$F^{Hodge}$ (eq.~\eqref{eq:transformation}). This is so, because, as it will be explained in Section~\ref{section: reformulation}, Theorem~\ref{theorem: main theorem} is a consequence of Theorem~\ref{theorem: DZ hierarchy}.
 
\subsection{Organization of the paper}

In Section~\ref{section: deformed kdv} we give a construction of the deformed KdV hierarchy. The main statement here is Proposition~\ref{proposition: deformed KdV}. 

In Section~\ref{section: CohFT} we recall the Dubrovin-Zhang theory of the hierarchies of the topological type. 

In Section~\ref{section: reformulation} we formulate three propositions and show that Theorems \ref{theorem: main theorem}, \ref{theorem: DZ hierarchy} and Proposition~\ref{proposition: deformed KdV} follow from them. These propositions are proved in Sections~\ref{section: pr1}, \ref{section: pr2} and \ref{section: pr3} correspondingly. 

In Section~\ref{section: ILW} we compare the deformed KdV hierarchy with the hierarchy of the Intermediate Long Wave equation.

Appendix is devoted to the proof of several technical statements.

\subsection{Acknowledgements} 

We would like to thank S.~Shadrin, B.~Dubrovin, H.~Posthuma, M.~Kazarian and R.~Pandharipande for useful discussions. We also thank the anonymous referee for valuable remarks and suggestions that allowed us to improve the exposition of this paper.

The author was supported by grant ERC-2012-AdG-320368-MCSK in the group of R. Pandharipande at ETH Zurich, by a Vidi grant of the Netherlands Organization for Scientific Research, Russian Federation Government grant no. 2010-220-01-077 (ag. no. 11.634.31.0005), the grants RFFI 13-01-00755, NSh-4850.2012.1, the Moebius Contest Foundation for Young Scientists and "Dynasty" foundation.


\section{Deformed KdV hierarchy}\label{section: deformed kdv}

In this section we construct the deformed KdV hierarchy. First, in Section~\ref{subsection: hamiltonian PDEs} we recall basic facts about hamiltonian systems of partial differential equations. Then in Section~\ref{subsection: deformed kdv} we present a construction of the deformed KdV hierarchy. The main statement here is Proposition~\ref{proposition: deformed KdV}. It says that there exists a unique sequence of local functionals with certain properties. The uniqueness part is simple. It is based on Lemma~\ref{lemma: uniqueness} that is proved in Section~\ref{subsection: proof of uniqueness}. The proof of the existence part is presented in Section~\ref{section: reformulation}.


\subsection{Hamiltonian systems of PDEs}\label{subsection: hamiltonian PDEs}

Here we recall the hamiltonian formalism in the theory of partial differential equations. The material of this section is mostly borrowed from \cite{DZ05}.

\subsubsection{Differential polynomials and local functionals}

Consider variables $u,u_1,u_2,\ldots$. We will often denote $u$ by $u_0$ and use an alternative notation for the variables $u_1,u_2,\ldots$:
$$
u_x:=u_1,\quad u_{xx}:=u_2,\ldots.
$$
Let $\cA$ be the space of polynomials in the variables $u_s, s=1,2,\ldots$,
$$
f(u;u_x,u_{xx},\ldots)=\sum_{m\ge 0}\sum_{s_1,\ldots,s_m\ge 1}f^{s_1,s_2,\ldots,s_m}(u)u_{s_1}u_{s_2}\ldots u_{s_m}
$$
with the coefficients $f^{s_1,\ldots,s_m}(u)$ being power series in $u$. Such an expression will be called differential polynomial.

The operator $\d_x\colon\cA\to\cA$ is defined as follows:
$$
\d_x:=\sum_{s\ge 0}u_{s+1}\frac{\d}{\d u_s}.
$$ 
Let $\Lambda=\cA/\im(\d_x)$. We have the projection $\pi\colon\cA\to\cA/\im(\d_x)$. We will use the following notation:
$$
\int h dx:=\pi(h),
$$
for any $h\in\cA$. The elements of the space $\Lambda$ will be called local functionals. 

For a local functional $\oh=\int hdx\in\Lambda$, the variational derivative $\frac{\delta\oh}{\delta u}\in\cA$ is defined as follows:
$$
\frac{\delta\overline h}{\delta u}:=\sum_{i\ge 0}(-\d_x)^i\frac{\d h}{\d u_i}.
$$ 

Let us introduce a gradation~$\deg_{dif}$ on the ring $\cA$ of differential polynomials putting 
$$
\deg_{dif} u_k=k,\, k\ge 1;\quad \deg_{dif} f(u)=0.
$$
This gradation will be called differential degree. The gradation on $\cA$ induces the gradation on the space $\Lambda$. There is an important lemma (see e.g. \cite{DZ05}).
\begin{lemma}\label{lemma: variational derivative}
Let $f$ be an arbitrary differential polynomial such that $f|_{u_i=0}=0$. Then the local functional $\overline f=\int f dx$ is equal to zero, if and only if $\frac{\delta\overline f}{\delta u}=0$. 
\end{lemma}

Let $\cA'\subset\cA$ be the subring of polynomials in $u,u_1,u_2,\ldots$. Sometimes we will use another gradation on the ring $\cA'$ assigning to $u_i, i\ge 0$, degree $1$. This second gradation will be just called degree. 

\subsubsection{Extended spaces}

Introduce a formal indeterminate $\hbar$ of the differential degree 
$$
\deg_{dif}\hbar=-2.
$$
Let $\widehat\cA:=\cA\otimes\mbC[[\hbar]]$ and $\hcA^{[k]}\subset\hcA$ be the subspace of elements of the total differential degree~$k$, $k\ge 0$. The space $\hcA^{[k]}$ consists of elements of the form
$$
f(u;u_1,u_2,\ldots;\hbar)=\sum_{i\ge 0}\hbar^i f_i(u;u_1,\ldots),\quad f_i\in\cA,\quad\deg_{dif}f_i=2i+k.
$$
The elements of the space $\hcA^{[k]}$ will be also called differential polynomials.

Let $\widehat\Lambda:=\Lambda\otimes\mbC[[\hbar]]$ and $\widehat\Lambda^{[k]}\subset\Lambda\otimes\mbC[[\hbar]]$ be the subspace of elements of the total differential degree~$k$. The space $\hLambda^{[k]}$ consists of integrals of the form
\begin{gather*}
\overline f=\int f(u;u_1,u_2,\ldots;\hbar)dx,\quad f\in\hcA^{[k]}.
\end{gather*}
They will also be called local functionals. 

\subsubsection{Hamiltonian systems of PDEs}

Let $K$ be a differential operator
\begin{gather}\label{eq: dif. operator}
K=\sum_{i,j\ge 0}f_{i,j}\hbar^i\d_x^j,
\end{gather}
where $f_{i,j}\in\cA$ and $\deg_{dif} f_{i,j}+j=2i+1$. Let us define the bracket $\{\cdot,\cdot\}_K\colon\widehat\Lambda^{[k]}\times\widehat\Lambda^{[l]}\to\widehat\Lambda^{[k+l+1]}$ by
\begin{gather*}
\{\overline g,\overline h\}_K:=\int\frac{\delta\overline g}{\delta u}K\frac{\delta\overline h}{\delta u}dx.
\end{gather*}

The operator $K$ is called Poisson, if the bracket $\{\cdot,\cdot\}_K$ is antisymmetric and satisfies the Jacobi identity. It is well-known that the operator $\d_x$ is Poisson (see e.g.~\cite{DZ05}).

A system of partial differential equations  
\begin{gather}\label{eq: system}
\frac{\d u}{\d t_i}=f_i(u;u_1,\ldots;\hbar),\quad i\ge 1,
\end{gather}
where $f_i\in\hcA^{[1]}$, is called hamiltonian, if there exists a Poisson operator $K$ and a sequence of local functionals $\overline h_i\in\widehat\Lambda^{[0]}$, $i\ge 1$, such that
\begin{align*}
&f_i=K\frac{\delta\oh_i}{\delta u},\\
&\{\oh_i,\oh_j\}_K=0,\quad\text{for $i,j\ge 1$}.
\end{align*}
The local functionals $\oh_i$ are called the Hamiltonians of the system \eqref{eq: system}. 

\subsubsection{Miura transformations}

Let us recall the Miura group action on hamiltonian hierarchies. 

Consider transformations of the form
\begin{gather}\label{eq: Miura group}
u\mapsto\widetilde u=u+\sum_{k\ge 1}\hbar^k f_k(u;u_1,\ldots,u_{2k}),\quad f_k\in\mathcal A,\quad \deg_{dif} f_k=2k.
\end{gather}
It is easy to see that transformations \eqref{eq: Miura group} form a group which is called the Miura group.

Let us define the Miura group action on hamiltonian hierarchies. Given a transformation~\eqref{eq: Miura group}, any differential polynomial from $\hcA^{[0]}$ can be rewritten in the variables $\widetilde u_i$. This defines the Miura group action on $\hcA^{[0]}$ and on $\widehat\Lambda^{[0]}$. The action on Poisson operators is defined as follows:
$$
K\mapsto\widetilde K=\left(\sum_{p\ge 0}\frac{\d\widetilde u}{\d u_p}\d_x^p\right)\circ K\circ\left(\sum_{q\ge 0}(-\d_x)^q\circ\frac{\d\widetilde u}{\d u_q}\right).
$$ 

The Miura group action transforms solutions of hamiltonian hierarchies in the following way (see e.g. \cite{DZ05}).
\begin{lemma}\label{lemma: transformed solution}
Suppose we have a Poisson operator $K$ and a sequence of commuting local functionals $\oh_n\in\hLambda^{[0]}$: $\{\oh_n,\oh_m\}_K=0$. Let $u(x,t_1,\ldots;\hbar)$ be a solution of the corresponding hierarchy of PDEs: $\frac{\d u}{\d t_n}=K\frac{\delta\oh_n}{\delta u}$. Consider a Miura transformation \eqref{eq: Miura group}. Then the series $\tu(x,t_1,\ldots;\hbar)$ is a solution of the transformed hierarchy: $\frac{\d\tu}{\d t_n}=\widetilde K\frac{\delta\oh_n}{\delta\tu}$.
\end{lemma}


\subsection{Deformed KdV hierarchy}\label{subsection: deformed kdv}

In this section we give a construction of a deformation of the KdV hierarchy.

\begin{proposition}\label{proposition: deformed KdV}
Let $\eps$ be any complex number. There exists a unique sequence of local functionals $\oh_n\in\hLambda^{[0]}, n\ge 1$, such that
\begin{align}
&\oh_1=\int\left(\frac{u^3}{6}+\sum_{g\ge 1}\hbar^g\eps^{g-1}\frac{|B_{2g}|}{2(2g)!}u u_{2g}\right)dx,\label{eq: h1}\\
&\oh_n=\int\left(\frac{u^{n+2}}{(n+2)!}+O(\hbar)\right)dx,\quad\text{for $n\ge 2$},\notag\\
&\{\oh_i,\oh_j\}_{\d_x}=0,\quad\text{for $i,j\ge 1$}.\notag
\end{align}
\end{proposition}

The hamiltonian system of partial differential equations corresponding to the sequence of local functionals $\oh_n$ and the Poisson operator $\d_x$ will be called the deformed KdV hierarchy.

The uniqueness statement in Proposition~\ref{proposition: deformed KdV} is a consequence of the following simple lemma that will be proved in the next section.
\begin{lemma}\label{lemma: uniqueness}
Let us fix a local functional $\oh\in\hLambda^{[0]}$ of the form $\oh=\int\left(\frac{u^3}{6}+O(\hbar)\right)dx$. Consider also an arbitrary power series $q_0(u)$. Suppose there exists a local functional $\overline q\in\hLambda^{[0]}$ of the form $\overline q=\int\left(q_0(u)+O(\hbar)\right)dx$, such that $\{\oh,\overline q\}_{\d_x}=0$. Then the local functional $\overline q$ is uniquely determined by $\oh$ and $q_0(u)$.
\end{lemma}
We thank B. Dubrovin for telling us about~Lemma~\ref{lemma: uniqueness}.

The proof of the existence part of Proposition~\ref{proposition: deformed KdV} is presented in Section~\ref{section: reformulation}.


\subsection{Proof of Lemma \ref{lemma: uniqueness}}\label{subsection: proof of uniqueness}

The proof is based on the following lemma.  
\begin{lemma}\label{lemma: tmp}
Let $p(u;u_1,u_2,\ldots)$ be an arbitrary homogeneous differential polynomial of positive differential degree. Suppose
$\left\{\int p dx,\int\frac{u^3}{6}dx\right\}_{\d_x}=0$, then $\int p dx=0$.
\end{lemma}
\begin{proof}
If $\deg_{dif} p=1$, then automatically $\int p dx=0$. Suppose $\deg_{dif} p\ge 2$. Define the bracket~$[\cdot,\cdot]$ on differential polynomials as follows:
$$
[q,r]:=\sum_{s\ge 0}\left((\d_x^s q)\frac{\d r}{\d u_s}-(\d_x^s r)\frac{\d q}{\d u_s}\right).
$$
We have 
\begin{multline*}
\int [uu_x,p] dx=\int\left(\sum_{s\ge 0}\d_x^s(u u_x)\frac{\d p}{\d u_s}\right)dx-\int(p u_x+u\d_x p)dx=\\
=\int\frac{\delta p}{\delta u}\d_x\left(\frac{u^2}{2}\right)dx-\int\d_x(p u)dx=\left\{\int p dx,\int\frac{u^3}{6}dx\right\}_{\d_x}=0.
\end{multline*}
Thus, $[uu_x,p]$ is a $\d_x$-derivative. 

Let us consider the lexicographical order on monomials $\prod_{k=1}^m u_k^{\alpha_k}$. It is easy to compute that, for a monomial $f(u)\prod_{k=1}^m u_k^{\alpha_k}$, we have (see \cite{LZ05})
\begin{gather}\label{eq: formula}
[uu_x,f(u)\prod_{k=1}^m u_k^{\alpha_k}]=\left(\sum_{k=1}^m(k+1)\alpha_k-\alpha_1-1\right)f(u)u_x\prod_{k=1}^m u_k^{\alpha_k}+
\begin{smallmatrix}
\text{monomials with the lower}\\
\text{lexicographical order}
\end{smallmatrix}.
\end{gather}

Let $f(u)\prod_{k=1}^m u_k^{\alpha_k}$ be the monomial in $p$ with the highest lexicographical order. From~\eqref{eq: formula} and the fact that $[u u_x,p]$ is a $\d_x$-derivative it follows that $m\ge 2$ and $\alpha_m=1$. The lexicographical order of the highest monomial in the polynomial
$$
p-\d_x\left(\frac{1}{\alpha_{m-1}+1}f(u)\left(\prod_{k=1}^{m-2} u_k^{\alpha_k}\right)u_{m-1}^{\alpha_{m-1}+1}\right)
$$
is lower than the lexicographical order of the highest monomial in $p$. We can do the same process further and prove that $p$ is a $\d_x$-derivative and, therefore, $\int p dx=0$.
\end{proof} 

Now let us prove Lemma~\ref{lemma: uniqueness}. Suppose that there exist two different local functionals $\overline q^1,\overline q^2\in\hLambda^{[0]}$, such that $\{\oh,\overline q^j\}_{\d_x}=0$ and $\overline q^j=\int\left(q_0(u)+\sum_{i\ge 1}q^j_i(u;u_1,\ldots)\hbar^i\right)dx$. We have 
\begin{gather}\label{eq: zero difference}
\{\oh,\overline q^1-\overline q^2\}_{\d_x}=0.
\end{gather}
Let $i_0$ be the smallest $i$, such that $\int(q_i^1-q_i^2)dx\ne 0$. From \eqref{eq: zero difference} it obviously follows that $\left\{\int\frac{u^3}{6}dx,\int(q^1_{i_0}-q^2_{i_0})dx,\right\}_{\d_x}=0$. Hence, by Lemma \ref{lemma: tmp}, $\int(q^1_{i_0}-q^2_{i_0})dx=0$. This contradiction proves the lemma.


\section{Cohomological field theories and the Dubrovin-Zhang hierarchies}\label{section: CohFT}

In this section we briefly recall the Dubrovin-Zhang theory of the hierarchies of the topological type. In Section~\ref{subsection: CohFT} we review the definition of cohomological field theory. In Section~\ref{subsection: DZ construction} we describe the construction of the Dubrovin-Zhang hierarchy associated to a semisimple cohomological field theory. 

\subsection{Cohomological field theory}\label{subsection: CohFT}

Here we recall the definition of cohomological field theory. For simplicity, we consider only one-dimensional cohomological field theories\footnote{To be completely precise, we consider one-dimensional cohomological field theories, where the scalar product of the unit with itself is equal to $1$.}. We refer the reader to \cite{Sha09} for a more detailed introduction to this subject.

A one-dimensional cohomological field theory is a collection of classes $\alpha_{g,n}\in H^*(\oM_{g,n};\mbC)$ defined for all $g$ and $n$ and satisfying the following properties (axioms):

\begin{itemize}

\item $\alpha_{g,n}$ belongs to the $S_n$-invariant part in the cohomology $H^*(\oM_{g,n};\mbC)$, where the $S_n$-action on $H^*(\oM_{g,n};\mbC)$ is induced by the mappings $\oM_{g,n}\to\oM_{g,n}$ defined by permutations of marked points.

\item We have $\alpha_{0,3}=1\in H^*(\oM_{0,3};\mbC)=\mbC$.

\item If $\pi\colon\oM_{g,n+1}\to \oM_{g,n}$ is the forgetful map, then $\pi^*\alpha_{g,n}=\alpha_{g,n+1}$.

\item  

\noindent a) If $gl\colon \oM_{g_1,n_1+1}\times\oM_{g_2,n_2+1}\to\oM_{g_1+g_2,n_1+n_2}$ is the gluing map, then $gl^*\alpha_{g_1+g_2,n_1+n_2}=\alpha_{g_1,n_1+1}\cdot\alpha_{g_2,n_2+1}.$

\noindent b) If $gl\colon\oM_{g-1,n+2}\to\oM_{g,n}$ is the gluing map, then $gl^*\alpha_{g,n}=\alpha_{g-1,n+2}$.

\end{itemize}

The potential $F$ of the cohomological field theory is defined as follows. Introduce variables~$t_d$, where $d\ge 0$. Then
\begin{align*}
&F:=\sum_{g\ge 0}F_g\hbar^g,\quad\text{where}\\
&F_g:=\sum_{\substack{n\ge 0\\2g-2+n>0}}\frac{1}{n!}\sum_{d_1,\ldots,d_n\ge 0}\left(\int_{\oM_{g,n}}\alpha_{g,n}\prod_{i=1}^n\psi_i^{d_i}\right)\prod_{i=1}^n t_{d_i}.
\end{align*}

\begin{example}
Let $\eps$ be an arbitrary complex number. Then the classes
$$
\alpha_{g,n}=1+\eps\lambda_1+\eps^2\lambda_2+\ldots+\eps^g\lambda_g\in H^*(\oM_{g,n};\mbC)
$$
form a one-dimensional cohomological field theory. 
\end{example}

\begin{example}
Let $\eps_1,\eps_2,\ldots$ be an arbitrary sequence of complex numbers. Then the classes
$$
\alpha_{g,n}=\exp\left(\sum_{i\ge 1}\eps_i\ch_{2i-1}(\Lambda)\right),
$$
where $\ch_{2i-1}(\Lambda)$ are the Chern characters of the Hodge bundle, form a one-dimensional cohomological field theory. In fact, any one-dimensional cohomological field theory has this form (see~\cite{MZ00}). 
\end{example}

\subsection{Dubrovin-Zhang hierarchy}\label{subsection: DZ construction}

In \cite{BPS12a} the authors gave a construction of a hamiltonian system of partial differential equations associated to an arbitrary semisimple cohomological field theory. In this section we recall that construction. For simplicity, we do it in the case of a one-dimensional cohomological field theory. Any one-dimensional cohomological field theory is semisimple.

We fix a one-dimensional cohomological field theory, $\alpha_{g,n}\in H^*(\oM_{g,n};\mbC)$, with a potential~$F=\sum_{g\ge 0}\hbar^g F_g$. In Sections~\ref{subsubsection: local functionals} and~\ref{subsubsection: Poisson operator} we construct a sequence of local functionals and a Poisson operator. In Section~\ref{subsection: solution} we present a solution of the constructed hierarchy. 

\subsubsection{Local functionals}\label{subsubsection: local functionals}

Let
$$
u:=\frac{\d^2 F}{\d t_0^2}.
$$ 
We identify $x$ with $t_0$. Let $u_n:=\d_x^n u$. From the axioms of cohomological field theory it follows that 
$$
u_n=t_n+\delta_{n,1}+O(t^2)+O(\hbar),\quad n\ge 0.
$$
Thus, any power series in $\hbar$ and $t_0,t_1,\ldots$ can be expressed as a power series in $\hbar$ and $u,u_1-1,u_2,u_3,\ldots$. 

Let 
$$
\Omega_{p,q}:=\frac{\d^2 F}{\d t_p\d t_q}.
$$ 
Let us express $\Omega_{p,q}$ as a power series in $\hbar$ and $u,u_1-1,u_2,\ldots$. In \cite{BPS12a} it is proved that the coefficient of $\hbar^g$ in $\Omega_{p,q}$ is a differential polynomial of differential degree $2g$. So, we can consider~$\Omega_{p,q}$ as an element of $\hcA^{[0]}$. Let $\oh_n:=\int \Omega_{0,n+1} dx\in\widehat\Lambda^{[0]}$, $n\ge 1$. The local functionals~$\oh_n$ will be the Hamiltonians of our hierarchy. It is easy to show that $\Omega_{0,n}=\frac{u^{n+1}}{(n+1)!}+O(\hbar)$.

\subsubsection{Poisson operator}\label{subsubsection: Poisson operator}

Let us construct a Poisson operator of our hierarchy. Let 
$$
v:=\frac{\d^2 F_0}{\d t_0^2}
$$
and $v_n:=\d_x^n v$. From the axioms of cohomological field theory it follows that 
$$
v_n=t_n+\delta_{n,1}+O(t^2).
$$
Thus, any power series in $t_0,t_1,t_2,\ldots$ can be expressed as a power series in $v,v_1-1,v_2,\ldots$. 

Consider $u$ as a power series in $v,v_1-1,v_2,\ldots$. Consider the differential operator
$$
K:=\left(\sum_{p\ge 0}\frac{\d u}{\d v_p}\d_x^p\right)\circ\d_x\circ\left(\sum_{q\ge 0}(-\d_x)^q\circ\frac{\d u}{\d v_q}\right).
$$
We can express this operator in the following form 
$$
K=\sum_{i,j\ge 0}p_{i,j}\hbar^i\d_x^j,
$$
where $p_{i,j}$ is a power series in $u,u_1-1,u_2,\ldots$. In \cite{BPS12a} it is proved that $p_{i,j}$ is a differential polynomial of differential degree $2i+1-j$. Thus, $K$ is an operator of the form \eqref{eq: dif. operator}. In fact, the operator~$K$ is Poisson and the local functionals $\oh_n$ commute with respect to the Poisson bracket defined by it: $\{\oh_n,\oh_m\}_K=0$.

By definition (see \cite{BPS12a}), the Dubrovin-Zhang hierarchy, associated to our cohomological field theory, is the hamiltonian hierarchy, formed by the local functionals $\oh_n, n\ge 1$, and the Poisson operator $K$. 

\subsubsection{Solution of the hierarchy}\label{subsection: solution}

We have the following lemma (see \cite{BPS12a}).
\begin{lemma}\label{lemma: solution}
The series $\frac{\d^2 F}{\d t_0^2}$ is a solution of the constructed hierarchy:
$$
\frac{\d u}{\d t_n}=K\frac{\delta\oh_n}{\delta u},\quad n\ge 1.
$$
\end{lemma}


\section{Reformulation of Theorems \ref{theorem: main theorem}, \ref{theorem: DZ hierarchy} and of Proposition~\ref{proposition: deformed KdV}}\label{section: reformulation}

In this section we formulate three propositions and show that Theorems \ref{theorem: main theorem}, \ref{theorem: DZ hierarchy} and Proposition~\ref{proposition: deformed KdV} follow from them. These propositions are proved in the next three sections of the paper.

Consider the cohomological field theory~\eqref{Hodge classes} and the corresponding Dubrovin-Zhang hierarchy.
\begin{proposition}\label{proposition: pr1}
The Miura transformation
\begin{gather}\label{eq: Miura transformation}
u\mapsto\widetilde u=u+\sum_{g\ge 1}\frac{(-1)^g}{2^{2g}(2g+1)!}\hbar^{g}\eps^g u_{2g}
\end{gather}
transforms the Poisson operator of the hierarchy to $\d_x$ and the Hamiltonian $\oh_1$ to 
\begin{gather}\label{eq: transformed hamiltonian}
\int\left(\frac{\tu^3}{6}+\frac{\hbar}{24}\tu\tu_2+\frac{\hbar^2\eps}{1440}\tu\tu_4+\sum_{g\ge 3}\hbar^g\eps^{g-1}c_g\tu\tu_{2g}\right)dx,
\end{gather}
where $c_g$, $g\ge 3$, are some complex constants.
\end{proposition}

\begin{proposition}\label{proposition: pr2}
The following two local functionals 
\begin{align*}
&\oh_1=\int\left(\frac{u^3}{6}+\sum_{g\ge 1}\hbar^g\eps^{g-1}\frac{|B_{2g}|}{2(2g)!}u u_{2g}\right)dx,\\
&\oh_2=\int\left(\frac{u^4}{4!}+\frac{\hbar}{48}u^2 u_{xx}+\sum_{g\ge 2}\frac{|B_{2g}|}{(2g)!}\hbar^g\left(\eps^{g-2}\frac{g+1}{2}u u_{2g}+\eps^{g-1}\frac{1}{4}u^2u_{2g}\right)\right)dx,
\end{align*}
commute with respect to the bracket $\{\cdot,\cdot\}_{\d_x}$.
\end{proposition}

\begin{proposition}\label{proposition: pr3}
Suppose there exists a sequence of complex numbers $c_g$, $g\ge 1$, $c_1\ne 0$, that satisfies the following property: there exists a local functional $\oh_2\in\hLambda^{[0]}$ of the form
$$
\oh_2=\int\left(\frac{u^4}{24}+O(\hbar)\right)dx
$$
that commutes with the local functional 
$$
\oh_1=\int\left(\frac{u^3}{6}+\sum_{g\ge 1}\hbar^{g}\eps^{g-1}c_g uu_{2g}\right)dx
$$
with respect to the bracket $\{\cdot,\cdot\}_{\d_x}$. Then all numbers $c_g$, for $g\ge 3$, are uniquely determined by $c_1$ and $c_2$.
\end{proposition}

Let us show that Theorems \ref{theorem: main theorem}, \ref{theorem: DZ hierarchy} and Proposition~\ref{proposition: deformed KdV} follow from these propositions. 

From the propositions it follows that the Miura transform of our Dubrovin-Zhang hierarchy is a hierarchy with $\d_x$ as a Poisson operator and the local functional \eqref{eq: h1} as the Hamiltonian~$\oh_1$. This proves the existence statement of Proposition~\ref{proposition: deformed KdV}. The uniqueness statement follows from Lemma~\ref{lemma: uniqueness}. We also immediately get Theorem~\ref{theorem: DZ hierarchy}. Theorem~\ref{theorem: main theorem} follows from Theorem~\ref{theorem: DZ hierarchy}, Lemma~\ref{lemma: solution} and Lemma~\ref{lemma: transformed solution}.


\section{Proof of Proposition \ref{proposition: pr1}}\label{section: pr1}

We have $\oh_n=\int\Omega_{0,n+1}dx$. The proof of the proposition is splitted in four steps. In Section~\ref{subsection: homogeneity} we derive a certain homogeneity property of the differential polynomials $\Omega_{p,q}$. In Section~\ref{subsection: coefficient of he} we find the coefficient of $\hbar^g\eps^g$ in the potential $F^{Hodge}$. In Section~\ref{subsection: Miura transformation} we prove that substitution~\eqref{eq: Miura transformation} kills the coefficients of $\hbar^g\eps^g$ in the Hamiltonians $\oh_n$ and show that 
$$
\oh_1=\int\left(\frac{\widetilde u^3}{6}+\sum_{g\ge 1}\hbar^g\eps^{g-1}c_g \tu\tu_{2g}\right)dx.
$$ 
We also show that $c_1=\frac{1}{24}$. The computation of $c_2$ is quite technical, it is done in Appendix~\ref{section: coefficient of h2}. Section~\ref{subsection: bracket} is devoted to the computation of the Poisson operator of our Dubrovin-Zhang hierarchy.

Let us fix some notations. By $F^{Hodge}$ we denote the potential of the cohomological field theory~\eqref{Hodge classes}. We also use the notations from Section~\ref{subsection: DZ construction}:
\begin{gather*}
u_i:=\d_x^i\frac{\d^2 F^{Hodge}}{\d t_0^2},\qquad v_i:=\d_x^i\frac{\d^2 F_0^{Hodge}}{\d t_0^2}.
\end{gather*}
Recall that we identify $x$ with $t_0$.

\subsection{Homogeneity of $\Omega_{p,q}$}\label{subsection: homogeneity}

The dimension of $\oM_{g,n}$ is equal to $3g-3+n$, thus, the coefficient of $\hbar^g\eps^j\prod_{i\ge 0}t_i^{d_i}$ in $F^{Hodge}$ is non-zero only if ${\sum_{i\ge 0}(i-1)d_i+j=3g-3}$. Consider the linear differential operator $O_1$ defined by 
$$
O_1:=\sum_{i\ge 0}(i-1)t_i\frac{\d}{\d t_i}+\eps\frac{\d}{\d\eps}-3\hbar\frac{\d}{\d\hbar}.
$$
We get 
\begin{gather}\label{eq: eq1}
O_1 F^{Hodge}=-3F^{Hodge}.
\end{gather}
From \eqref{eq: eq1} and the commutation relation (recall that $\d_x=\frac{\d}{\d t_0}$)
\begin{gather}\label{eq: commutation relation}
[\d_x,O_1]=-\d_x
\end{gather}
it is clear that
$$
O_1u_n=(n-1)u_n.
$$
Thus, 
\begin{gather*}
O_1=\sum_{i\ge 0}(i-1)u_i\frac{\d}{\d u_i}+\eps\frac{\d}{\d\eps}-3\hbar\frac{\d}{\d\hbar}.
\end{gather*}

From \eqref{eq: eq1} it is easy to see that
\begin{gather}\label{eq: eq2}
O_1\Omega_{p,q}=-(p+q+1)\Omega_{p,q}.
\end{gather}
On the other hand, in \cite{BPS12a} it is proved that $\Omega_{p,q}$ is a power series in $\hbar$, where the coefficient of $\hbar^g$ is a homogeneous differential polynomial of differential degree $2g$. This property can be written as 
\begin{align}
&O_2\Omega_{p,q}=0,\quad\text{where}\label{eq: eq3}\\
&O_2:=\sum_{i\ge 0}iu_i\frac{\d}{\d u_i}-2\hbar\frac{\d}{\d\hbar}.\notag
\end{align}
If we subtract \eqref{eq: eq2} from \eqref{eq: eq3}, we get
\begin{gather}\label{eq: homogeneity}
\left(\sum_{i\ge 0}u_i\frac{\d}{\d u_i}+\hbar\frac{\d}{\d\hbar}-\eps\frac{\d}{\d\eps}\right)\Omega_{p,q}=(p+q+1)\Omega_{p,q}.
\end{gather}
We have that $\Omega_{p,q}$ is a power series in $\hbar$ and $\eps$ with the coefficients that are differential polynomials. It is easy to see that the coefficient of $\hbar^g\eps^j$ is non-zero, only if $g\ge j$. From \eqref{eq: homogeneity} it follows that the coefficient of $\hbar^g\eps^j$ is a polynomial in $u,u_1,\ldots$ of degree $p+q+1-g+j$.

\subsection{Coefficient of $\hbar^g\eps^g$}\label{subsection: coefficient of he}

The so-called $\lambda_g$-conjecture, proved in \cite{FP03}, tells that
\begin{gather}\label{eq: lambdag}
\int_{\oM_{g,n}}\lambda_g\psi_1^{d_1}\psi_2^{d_2}\ldots\psi_n^{d_n}=\frac{2^{2g-1}-1}{2^{2g-1}}\frac{|B_{2g}|}{(2g)!}\frac{(2g-3+n)!}{d_1!d_2!\ldots d_n!},\quad g\ge 1,\quad \sum_{i=1}^n d_i=2g-3+n.
\end{gather}
We have 
$$
F^{Hodge}_0=\sum_{n\ge 3}\frac{1}{n!}\sum_{\substack{d_1,\ldots,d_n\ge 0\\d_1+\ldots+d_n=n-3}}\frac{(n-3)!}{d_1!\ldots d_n!}t_{d_1}\ldots t_{d_n}.
$$
Therefore, from \eqref{eq: lambdag} it follows that, for $g\ge 1$, the coefficient of $\hbar^g\eps^g$ in $F^{Hodge}$ is equal to~$\frac{2^{2g-1}-1}{2^{2g-1}}\frac{|B_{2g}|}{(2g)!}v_{2g-2}$.

Consider now $\Omega_{0,n}=\frac{\d^2 F^{Hodge}}{\d t_0\d t_n}$ as a series in $\hbar,\eps,v,v_1-1,v_2,\ldots$. We get that the coefficient of $\hbar^g\eps^g$ is equal to
$$
\frac{2^{2g-1}-1}{2^{2g-1}}\frac{|B_{2g}|}{(2g)!}\d_x^{2g-1}\left(\frac{\d v}{\d t_n}\right)=\frac{2^{2g-1}-1}{2^{2g-1}}\frac{|B_{2g}|}{(2g)!}\d_x^{2g-1}\left(\frac{v^n}{n!}v_x\right)=\frac{2^{2g-1}-1}{2^{2g-1}}\frac{|B_{2g}|}{(2g)!}\d_x^{2g}\left(\frac{v^{n+1}}{(n+1)!}\right).
$$ 
Thus,
$$
\Omega_{0,n}=\frac{v^{n+1}}{(n+1)!}+\sum_{g\ge 1}\hbar^g\eps^g\frac{2^{2g-1}-1}{2^{2g-1}}\frac{|B_{2g}|}{(2g)!}\d_x^{2g}\left(\frac{v^{n+1}}{(n+1)!}\right)+\sum_{g>j\ge 0}\hbar^g\eps^jf^n_{g,j}(v,v_1-1,v_2,\ldots),
$$
where $f^n_{g,j}(v,v_1-1,v_2,\ldots)$ are power series in $v,v_1-1,v_2,\ldots$.

\subsection{Miura transformation}\label{subsection: Miura transformation}

We have 
\begin{gather}\label{eq: formula for u}
u=v+\sum_{g\ge 1}(\hbar\eps)^g\frac{2^{2g-1}-1}{2^{2g-1}}\frac{|B_{2g}|}{(2g)!}v_{2g}+\sum_{g>j\ge 0}\hbar^g\eps^j f^0_{g,j}(v,v_1-1,v_2,\ldots).
\end{gather}
It is easy to check that 
\begin{gather}\label{eq:identity for Bernoulli}
\left(1+\sum_{g\ge 1}\frac{2^{2g-1}-1}{2^{2g-1}}\frac{|B_{2g}|}{(2g)!}z^{2g}\right)\left(1+\sum_{g\ge 1}\frac{(-1)^g}{2^{2g}(2g+1)!}z^{2g}\right)=1.
\end{gather}
Therefore,
$$
v=\tu+\sum_{g>j\ge 0}\hbar^g\eps^jq_{g,j}(\tu,\tu_1-1,\tu_2,\ldots),
$$
where $q_{g,j}(\tu,\tu_1-1,\tu_2,\ldots)$ are power series in $\tu,\tu_1-1,\tu_2,\ldots$.

We get
\begin{align*}
&\Omega_{0,n}=\frac{\tu^{n+1}}{(n+1)!}+\sum_{g\ge 1}\hbar^g\eps^g\frac{2^{2g-1}-1}{2^{2g-1}}\frac{|B_{2g}|}{(2g)!}\d_x^{2g}\left(\frac{\tu^{n+1}}{(n+1)!}\right)+\sum_{g>j\ge 0}\hbar^g\eps^jw^n_{g,j}(\tu,\tu_1,\ldots),\\
&\oh_n=\int\Omega_{0,n+1}dx=\int\left(\frac{\tu^{n+2}}{(n+2)!}+\sum_{g>j\ge 0}\hbar^g\eps^j w^{n+1}_{g,j}(\tu,\tu_1,\ldots)\right)dx.
\end{align*}
Here $w^n_{g,j}$ are differential polynomials in $\tu_i$. From \eqref{eq: homogeneity} it follows that $w^n_{g,j}$ is a polynomial in~$\tu,\tu_1,\ldots$ of degree $n+1-g+j$. If $g-j=n$, then $w^n_{g,j}=b^n_g\tu_{2g}$, for some constant $b^n_g$ and we have $\int w^n_{g,j} dx=0$. We obtain
$$
\oh_n=\int\left(\frac{\tu^{n+2}}{(n+2)!}+\sum_{\substack{g,j\ge 0\\n\ge g-j\ge 1}}\hbar^g\eps^j w^{n+1}_{g,j}\right)dx.
$$
In particular, we get
$$
\oh_1=\int\left(\frac{\widetilde u^3}{6}+\sum_{g\ge 1}\hbar^g\eps^{g-1}q_g\right)dx,
$$
where $q_g$ are quadratic polynomials in $\tu,\tu_1,\ldots$. It is cleat that $\int \tu_i\tu_jdx=(-1)^i\int \tu\tu_{i+j}dx$. Therefore, we have
\begin{gather}\label{eq: unknown coefficients}
\oh_1=\int\left(\frac{\widetilde u^3}{6}+\sum_{g\ge 1}\hbar^g\eps^{g-1}c_g \tu\tu_{2g}\right)dx,
\end{gather}
for some constants $c_g$. 

It remains to prove that $c_1=\frac{1}{24}$ and $c_2=\frac{1}{1440}$. If $\eps=0$, then our cohomological field theory is trivial. The corresponding Dubrovin-Zhang hierarchy in this case is the KdV hierarchy (see~\cite{DZ05}). Thus, $c_1=\frac{1}{24}$. The computation of $c_2$ is done in Appendix~\ref{section: coefficient of h2}.

\subsection{Poisson operator}\label{subsection: bracket}

Consider the operator $O_1$ from Section~\ref{subsection: homogeneity}. Since $O_1v=-v$ and $O_1v_n=(n-1)v_n$, we get
$$
O_1=\sum_{i\ge 0}(i-1)v_i\frac{\d}{\d v_i}+\eps\frac{\d}{\d\eps}-3\hbar\frac{\d}{\d\hbar}.
$$
Thus,
\begin{gather}\label{eq: degree}
O_1\frac{\d u}{\d v_n}=-n\frac{\d u}{\d v_n}.
\end{gather}

The Poisson operator $K$ of our hierarchy is equal to  
\begin{gather}\label{eq: formula for K}
K=\left(\sum_{m\ge 0}\frac{\d u}{\d v_m}\d_x^m\right)\circ\d_x\circ\left(\sum_{n\ge 0}(-\d_x)^n\circ\frac{\d u}{\d v_n}\right).
\end{gather}
Let us express it as $K=\sum_{n\ge 0}p_n\d_x^n$. From \eqref{eq: degree} and \eqref{eq: commutation relation} it follows that 
\begin{gather}\label{eq: bracket1}
O_1p_n=-(n-1)p_n.
\end{gather}
On the other hand, in \cite{BPS12a} it is proved that the coefficient of $\hbar^g\d_x^n$ in $K$ is a differential polynomial in $u,u_1,\ldots$ of differential degree $2g+1-n$. Therefore, we have
\begin{gather}\label{eq: bracket2}
\left(-\sum_{i\ge 0}iu_i\frac{\d}{\d u_i}+2\hbar\frac{\d}{\d\hbar}\right)p_n=(n-1)p_n.
\end{gather}

Let us sum \eqref{eq: bracket1} and \eqref{eq: bracket2}, we get
\begin{gather}\label{eq: bracket homogeneity}
\left(-\sum_{i\ge 0}u_i\frac{\d}{\d u_i}+\eps\frac{\d}{\d\eps}-\hbar\frac{\d}{\d\hbar}\right)p_n=0.
\end{gather}
We know that $p_n$ is a power series in $\hbar$ and $\eps$ with the coefficients that are differential polynomials in $u_i$. It is easy to see that the coefficient of $\hbar^g\eps^j$ is zero, if $g<j$. Thus, from~\eqref{eq: bracket homogeneity} it follows that $p_n=\sum_{g\ge 0}b_{g,n}\hbar^g\eps^g$, where $b_{g,n}$ are complex numbers. From \eqref{eq: bracket2} it follows that $b_{g,n}=0$, if $2g\ne n-1$. Finally, we get
\begin{gather}\label{eq: bracket3}
K=\sum_{g\ge 0}b_g\hbar^g\eps^g\d_x^{2g+1},
\end{gather}
where $b_g$ are some complex numbers.

We have proved that in the operator $K$ there are no terms with $\hbar^g\eps^j$, for $g>j$. Thus, by~\eqref{eq: formula for u} and \eqref{eq: formula for K},
\begin{gather}\label{eq:operator K as a product}
K=\left[1+\sum_{g\ge 1}(\hbar\eps)^g\frac{2^{2g-1}-1}{2^{2g-1}}\frac{|B_{2g}|}{(2g)!}\d_x^{2g}\right]\circ\d_x\circ\left[1+\sum_{g\ge 1}(\hbar\eps)^g\frac{2^{2g-1}-1}{2^{2g-1}}\frac{|B_{2g}|}{(2g)!}\d_x^{2g}\right].
\end{gather}
This equation together with~\eqref{eq:identity for Bernoulli} implies that the Miura transformation~\eqref{eq: Miura transformation} transforms the operator $K$ to $\d_x$. This concludes the proof of the proposition.

\begin{remark}
Let us compute the product on the right-hand side of~\eqref{eq:operator K as a product}. By~\eqref{eq:identity for Bernoulli}, 
$$
1+\sum_{g\ge 1}\frac{2^{2g-1}-1}{2^{2g-1}}\frac{|B_{2g}|}{(2g)!}z^{2g}=\frac{iz}{e^{\frac{iz}{2}}-e^{-\frac{iz}{2}}}.
$$
Let $\psi(z):=\frac{iz}{2}\frac{e^{\frac{iz}{2}}+e^{-\frac{iz}{2}}}{e^{\frac{iz}{2}}-e^{-\frac{iz}{2}}}$. A direct computation shows that
$$
\left(\frac{iz}{e^{\frac{iz}{2}}-e^{-\frac{iz}{2}}}\right)^2=\psi-z\psi'.
$$
On the other hand, $\psi(z)=1-\sum_{g\ge 1}\frac{|B_{2g}|}{(2g)!}z^{2g}$. Therefore,
$$
\psi-z\psi'=1+\sum_{g\ge 1}\frac{(2g-1)|B_{2g}|}{(2g)!}z^{2g}.
$$
We conclude that
\begin{gather*}
K=\d_x+\sum_{g\ge 1}\hbar^g\eps^g\frac{(2g-1)|B_{2g}|}{(2g)!}\d_x^{2g+1}.
\end{gather*}
\end{remark}


\section{Proof of Proposition \ref{proposition: pr2}}\label{section: pr2}

Before the proof of the proposition let us state several useful formulas.
\begin{align}
&\left\{\int u u_{2g_1}dx,\int u u_{2g_2}dx\right\}_{\d_x}=0,\label{eq: identity1}\\
&\left\{\int\frac{u^3}{6}dx,\int u^2u_{2g}dx\right\}_{\d_x}=-2\left\{\int u u_{2g}dx,\int\frac{u^4}{24}dx\right\}_{\d_x}.\label{eq: identity2}
\end{align}
They can be easily checked by a direct computation.

From \eqref{eq: identity1} and \eqref{eq: identity2} it follows that
\begin{align}
&\left\{\oh_1,\oh_2\right\}_{\d_x}=\sum_{g\ge 2}\hbar^g\eps^{g-2}\times\notag\\
&\times\left[\frac{(g+1)|B_{2g}|}{(2g)!}\left\{\int\frac{u^3}{6}dx,\int\frac{uu_{2g}}{2}dx\right\}_{\d_x}+\sum_{\substack{g_1+g_2=g\\g_1,g_2\ge 1}}\frac{|B_{2g_1}||B_{2g_2}|}{8(2g_1)!(2g_2)!}\left\{\int uu_{2g_1}dx,\int u^2u_{2g_2}dx\right\}_{\d_x}\right].\label{eq: integral}
\end{align}
We have to prove that \eqref{eq: integral} is equal to $0$. Expression~\eqref{eq: integral} is equal to
\begin{gather}\label{eq: integral2}
\int\left[\frac{(g+1)|B_{2g}|}{2(2g)!}u^2u_{2g+1}-\sum_{\substack{g_1+g_2=g\\g_1,g_2\ge 1}}\frac{|B_{2g_1}||B_{2g_2}|}{4(2g_1)!(2g_2)!}(2uu_{2g_2}+\d_x^{2g_2}(u^2))u_{2g_1+1}\right]dx.
\end{gather}
We have $\int\d_x^{2g_2}(u^2)u_{2g_1+1}dx=\int u^2u_{2g+1}$. If $m\ge 2$, then (see e.g. \cite{GKP94}) 
$$
\sum_{\substack{m_1+m_2=m\\m_1,m_2\ge 1}}\frac{|B_{2m_1}||B_{2m_2}|}{(2m_1)!(2m_2)!}=\frac{(2m+1)|B_{2m}|}{(2m)!}.
$$
Therefore,~\eqref{eq: integral2} is equal to
$$
\int\left[\frac{|B_{2g}|}{4(2g)!}u^2u_{2g+1}-\sum_{\substack{g_1+g_2=g\\g_1,g_2\ge 1}}\frac{|B_{2g_1}||B_{2g_2}|}{2(2g_1)!(2g_2)!}u u_{2g_2}u_{2g_1+1}\right]dx.
$$ 
The variational derivative of this integral is equal to
\begin{align*}
&\frac{|B_{2g}|}{4(2g)!}(2u u_{2g+1}-\d_x^{2g+1}(u^2))-\sum_{\substack{g_1+g_2=g\\g_1,g_2\ge 1}}\frac{|B_{2g_1}||B_{2g_2}|}{2(2g_1)!(2g_2)!}(u_{2g_2}u_{2g_1+1}+\d_x^{2g_2}(u u_{2g_1+1})-\d_x^{2g_1+1}(u u_{2g_2}))\\
&=\frac{|B_{2g}|}{4(2g)!}(2u u_{2g+1}-\d_x^{2g+1}(u^2))-\sum_{\substack{g_1+g_2=g\\g_1,g_2\ge 1}}\frac{|B_{2g_1}||B_{2g_2}|}{2(2g_1)!(2g_2)!}(u_{2g_2}u_{2g_1+1}-\d_x^{2g_1}(u_1 u_{2g_2}))\\
&=\frac{(-1)^{g+1}}{2}\sum_{i=0}^g\frac{B_{2i}B_{2g-2i}}{(2i)!(2g-2i)!}(u_{2i}u_{2g-2i+1}-\d_x^{2i}(u_1 u_{2g-2i})).
\end{align*}
\begin{lemma}\label{lemma: identity}
We have the following identity:
\begin{gather}\label{eq: Bernoulli identity}
\sum_{i=0}^g\frac{B_{2i}B_{2g-2i}}{(2i)!(2g-2i)!}(u_{2i}u_{2g-2i+1}-\d_x^{2i}(u_1 u_{2g-2i}))=
\begin{cases}
-\frac{u_1u_2}{4},&\text{if $g=1$},\\
0,&\text{if $g\ne 1$}.
\end{cases}
\end{gather}
\end{lemma}
\noindent This lemma concludes the proof of Proposition~\ref{proposition: pr2}. We prove it in Appendix~\ref{section: proof of the identity}.


\section{Proof of Proposition \ref{proposition: pr3}}\label{section: pr3}

We have
\begin{gather}
\oh_2=\int\left(\frac{u^4}{24}+\sum_{g\ge 1}\hbar^g p_g\right)dx. 
\end{gather}
It is easy to see that $\int p_1 dx=\int\frac{c_1}{2}u^2u_2dx$. Denote $\frac{c_1}{2}u^2u_2$ by $r_1$. 

Let us show that, for $g\ge 2$, we have 
\begin{gather}\label{eq: degrees}
\int p_g dx=\int\left(\eps^{g-2}q_g+\eps^{g-1}r_g\right)dx,
\end{gather}
where $q_g$ and $r_g$ are polynomials in $u_i$ of degrees $2$ and $3$ correspondingly. We prove it by induction on $g$. The coefficient of $\hbar^g$ in $\{\oh_1,\oh_2\}_{\d_x}$ is equal to
\begin{multline}
\left\{\int\frac{u^3}{6}dx,\int p_gdx\right\}_{\d_x}+\eps^{g-2}\sum_{\substack{g_1+g_2=g\\g_1,g_2\ge 1}}c_{g_1}\left\{\int uu_{2g_1}dx,\int r_{g_2}dx\right\}_{\d_x}+\\
+\eps^{g-1}c_g\left\{\int uu_{2g}dx,\int\frac{u^4}{24}dx\right\}_{\d_x}=0.\label{eq: sum}
\end{multline}
The second term in \eqref{eq: sum} has degree $3$ and the third one has degree $4$. Hence, we get~\eqref{eq: degrees}.

From \eqref{eq: sum} and \eqref{eq: identity2} it follows that $\int r_gdx=\frac{c_g}{2}\int u^2u_{2g}dx$. Clearly, we have $\int q_g dx=e_g\int uu_{2g}dx$, where $e_g$ is a complex constant. Using~\eqref{eq: sum}, we get
\begin{gather}\label{eq: coefficients}
e_g\left\{\int\frac{u^3}{6}dx,\int uu_{2g}dx\right\}_{\d_x}+\sum_{\substack{g_1+g_2=g\\g_1,g_2\ge 1}}\frac{c_{g_1}c_{g_2}}{2}\left\{\int uu_{2g_1}dx,\int u^2u_{2g_2}dx\right\}_{\d_x}=0.
\end{gather}

Define the local functionals $\overline f_g,\overline f_{g_1,g_2}\in\Lambda^{[0]}$ as follows:
\begin{align*}
\overline f_g&:=\left\{\int\frac{u^3}{6}dx,\int uu_{2g}dx\right\}_{\d_x},\\
\overline f_{g_1,g_2}&:=\left\{\int uu_{2g_1}dx,\int u^2u_{2g_2}dx\right\}_{\d_x}+\left\{\int uu_{2g_2}dx,\int u^2u_{2g_1}dx\right\}_{\d_x}.
\end{align*}
In these notations equation~\eqref{eq: coefficients} looks as follows:
$$
e_g\overline f_g+\frac{c_{g-1}c_1}{2}\overline f_{g-1,1}+\sum_{\substack{g_1+g_2=g\\g_1\ge g_2\ge 2}}\frac{c_{g_1}c_{g_2}}{2}\overline f_{g_1,g_2}=0.
$$

Let us show that, for $g\ge 4$, this equation uniquely determines $c_{g-1}$ from $c_{g-2},c_{g-3},\ldots,c_1$. For this we have to prove that the local functionals $\overline f_g$ and $\overline f_{g-1,1}$ are linearly independent. We have
\begin{align*}
\of_g=&\int u^2u_{2g+1}dx,\\
\of_{g-1,1}=&\int\left[-2(\d_x^2(u^2)+2uu_2)u_{2g-1}-2(2uu_{2g-2}+\d_x^{2g-2}(u^2))u_3\right]dx=\\
=&-4\of_g-2\int(2uu_2u_{2g-1}+2uu_3u_{2g-2})dx=\\
=&-4\of_g-2\int(\d_x^2(u^2)u_{2g-1}-2u_1^2u_{2g-1}+\d_x^3(u^2)u_{2g-2}-6u_1u_2u_{2g-2})dx=\\
=&-4\of_g-2\int u_1^2u_{2g-1}dx.
\end{align*}
We need to prove that $\frac{\delta}{\delta u}\int(u^2u_{2g+1})dx$ and $\frac{\delta}{\delta u}\int(u_x^2u_{2g-1})dx$ are linearly independent. We have
\begin{align}
&\frac{\delta}{\delta u}\int(u^2u_{2g+1})dx=-2\sum_{i=1}^g{2g+1\choose i}u_iu_{2g+1-i},\label{eq: var1}\\
&\frac{\delta}{\delta u}\int(u_x^2u_{2g-1})dx=-2u_1u_{2g}-2u_2u_{2g-1}-2\sum_{i=1}^g{2g-1\choose i-1}u_iu_{2g+1-i}.\label{eq: var2}
\end{align}
The matrix of coefficients of $u_1u_{2g}$ and $u_3u_{2g-2}$ in \eqref{eq: var1} and \eqref{eq: var2} is equal to
$$
\begin{pmatrix}
-2(2g+1) & -\frac{(2g+1)2g(2g-1)}{3}\\
-4 & -(2g-1)(2g-2) 
\end{pmatrix}
$$
It is non-degenerate, if $g\ge 4$. This completes the proof of the proposition.


\section{Deformed KdV hierarchy and the ILW equation}\label{section: ILW}

In this section we explain a relation of the deformed KdV hierachy to the hierarchy of the conserved quantities of the Intermediate Long Wave equation constructed in~\cite{SAK79}. 

In Section \ref{subsection: ILW equation} we recall the definition of the ILW equation and show how to rescale the parameters in order to get the first equation \eqref{eq:ILW equation} of the deformed KdV hierarchy. In Section~\ref{subsection: extensions} we introduce slight extensions of the spaces $\hcA^{[k]}$ and $\hLambda^{[k]}$. Section~\ref{subsection: conserved quantities} contains a review of the construction of conserved quantities of the ILW equation from \cite{SAK79}. In Section~\ref{subsection: relation} we compare these conserved quantities with the Hamiltonians of the deformed KdV hierarchy. 

\subsection{Intermediate Long Wave equation}\label{subsection: ILW equation}

The Intermediate Long Wave equation looks as follows (see e.g. \cite{SAK79}):
\begin{gather}\label{eq: ILW}
w_\tau+2w w_x+T(w_{xx})=0,
\end{gather}
where 
$$
T(f):=\sum_{n\ge 1}\delta^{2n-1}2^{2n}\frac{|B_{2n}|}{(2n)!}\d_x^{2n-1} f
$$
and $\delta$ is a non-zero complex number.
\begin{remark}
In the physics literature the operator $T$ is usually written in the following way:
$$
T(f)=PV\int_{-\infty}^{\infty}\frac{1}{2\delta}\left(\sgn(x-\xi)-\coth\frac{\pi(x-\xi)}{2\delta}\right)f(\xi)d\xi.
$$
\end{remark}

Let $\mu$ be a formal variable and $\eps$ be a non-zero complex number. Let us make the following rescalings:
\begin{gather}\label{rescallings}
w=\frac{\sqrt{\eps}}{\mu}u,\qquad \tau=-\frac{\mu}{2\sqrt{\eps}}t,\qquad\delta=\frac{\mu\sqrt{\eps}}{2}.
\end{gather}
Then equation~\eqref{eq: ILW} is transformed to
\begin{gather}\label{eq:our ILW}
u_t=u u_x+\sum_{g\ge 1}\mu^{2g}\eps^{g-1}\frac{|B_{2g}|}{(2g)!}u_{2g+1}.
\end{gather}
If we put $\hbar=\mu^2$, we get exactly the first equation~\eqref{eq:ILW equation} of the deformed KdV hierarchy.


\subsection{Extensions of $\hcA^{[k]}$ and of $\hLambda^{[k]}$}\label{subsection: extensions}

We need to enlarge the spaces $\hcA^{[k]}$ and $\hLambda^{[k]}$. 

Let $\hcA_\mu^{[k]}$ be the space of series of the form
$$
f(u;u_1,u_2,\ldots;\mu)=\sum_{i\ge 0}\mu^i f_i(u;u_1,\ldots),\quad f_i\in\cA,\quad\deg_{dif}f_i=i+k.
$$
Denote by $\hLambda^{[k]}_\mu$ the space of integrals of the form
\begin{gather*}
\overline f=\int f(u;u_1,u_2,\ldots;\mu)dx,\quad \text{where $f\in\hcA_\mu^{[k]}$}.
\end{gather*}

We have the following simple generalization of Lemma~\ref{lemma: uniqueness}.
\begin{lemma}\label{lemma:mu-uniqueness}
Let us fix a local functional $\oh\in\hLambda^{[0]}_\mu$ of the form $\oh=\int\left(\frac{u^3}{6}+O(\mu)\right)dx$. Consider also an arbitrary power series $q_0(u)$. Suppose there exists a local functional $\overline q\in\hLambda^{[0]}_\mu$ of the form $\overline q=\int\left(q_0(u)+O(\mu)\right)dx$, such that $\{\oh,\overline q\}_{\d_x}=0$. Then the local functional $\overline q$ is uniquely determined by $\oh$ and $q_0(u)$.
\end{lemma}
The proof is the same as the proof of Lemma~\ref{lemma: uniqueness}.


\subsection{Conserved quantities}\label{subsection: conserved quantities}

Here we review the construction of an infinite sequence of conserved quantities of the ILW equation. We follow \cite{SAK79} except for the fact that we make the rescalings~\eqref{rescallings}.

Let us introduce the operator $R$ by
$$
R:=\sum_{g\ge 1}\mu^{2g-1}\eps^{g-1}\frac{|B_{2g}|}{(2g)!}\d_x^{2g-1}.
$$
Consider the following equation:
$$
e^\sigma-1=\frac{1}{\lambda}\left(\frac{2}{\eps}\sigma-\mu\left(\frac{i}{\sqrt{\eps}}+2 R\right)\sigma_x+2u\right).
$$
It is easy to see that it has a unique solution of the form $\sigma=\sum_{n\ge 1}\frac{\sigma_n}{\lambda^n}$, where $\sigma_n\in\hcA^{[0]}_\mu$. For example,
\begin{align*}
\sigma_1&=2u,\\
\sigma_2&=-2u^2-2\mu\left(\frac{i}{\sqrt{\eps}}+2R\right)u_x+\frac{4u}{\eps}.
\end{align*}
It is not hard to check that, if $u$ is a solution of \eqref{eq:our ILW}, then $\sigma$ satisfies the following equation:
\begin{gather*}
\sigma_t=\frac{\lambda}{2}\left(e^\sigma-1\right)\sigma_x-\eps^{-1}\sigma\sigma_x+\mu\sigma_x R\sigma_x+\mu R\sigma_{xx}.
\end{gather*}
We can easily see that $\int\sigma_t dx=0$, therefore, all local functionals $\int\sigma_n dx$ are conserved quantities of the equation~\eqref{eq:our ILW}.  


\subsection{Relation to the deformed KdV hierarchy}\label{subsection: relation}

In this section we express the conserved quantities $\int\sigma_n dx$ as linear combinations of the Hamiltonians $\oh_n$.

Let $\hbar=\mu^2$ and consider the Hamiltonians $\oh_n, n\ge 1$, of the deformed KdV hierarchy. Let $\oh_0:=\int\frac{u^2}{2}dx$ and $\oh_{-1}:=\int u dx$. It is easy to see that 
$$
\sigma_n=(-1)^{n+1}\frac{2^n}{n}u^n+\sum_{i=1}^{n-1}\frac{a_{i,n}}{(n-i)!}\frac{u^{n-i}}{\eps^i}+O(\mu),
$$
where $a_{i,j}, 1\le i<j$, are some complex coefficients. Thus, we have
$$
\int\sigma_n dx=(-1)^{n+1}2^n(n-1)!\oh_{n-2}+\sum_{i=1}^{n-1}\frac{a_{i,n}}{\eps^i}\oh_{n-i-2}+O(\mu).
$$
Since $\int\sigma_n dx$ are conserved quantities, we have $\{\int\sigma_n dx,\oh_1\}_{\d_x}=0$. Therefore, from Lemma~\ref{lemma:mu-uniqueness} it follows that
$$
\int\sigma_n dx=(-1)^{n+1}2^n(n-1)!\oh_{n-2}+\sum_{i=1}^{n-1}\frac{a_{i,n}}{\eps^i}\oh_{n-i-2}.
$$ 


\appendix

\section{Coefficient of $\hbar^2$}\label{section: coefficient of h2}

Here we compute the coefficient $c_2$ in~\eqref{eq: unknown coefficients} and complete the proof of Proposition~\ref{proposition: pr1}. 

Consider the local functionals $\int\Omega_{0,2}dx$ before the Miura transformation~\eqref{eq: Miura transformation}. In order to compute the coefficient $c_2$ in~\eqref{eq: unknown coefficients}, we only need to compute the coefficients of $\hbar$ and of $\hbar^2\eps$ in~$\int\Omega_{0,2}dx$. The coefficient of $\hbar$ is equal to $\frac{1}{24}\int u u_2dx$. Let us compute the coefficient of $\hbar^2\eps$.  

The series $\left.\frac{\d\Omega_{0,2}}{\d\eps}\right|_{\eps=0}$ can be computed using the Givental operators that act on potentials of cohomological field theories. We remind the general formulas for this in Section~\ref{subsection: deformations}. All technical computations are done in Section~\ref{subsection: coefficient of h2}.   

\subsection{Deformations of cohomological field theories}\label{subsection: deformations}

Consider a one-dimensional cohomological field theory, $\alpha_{g,n}\in H^*(\oM_{g,n};\mbC)$. Let $F$ be its potential. Consider the following deformation of the classes $\alpha_{g,n}$:
$$
\alpha_{g,n}(\eps)=\exp\left(\eps\ch_{2l-1}(\Lambda)\right)\alpha_{g,n},
$$
where $\ch_{2l-1}(\Lambda)$ is the Chern character of the Hodge bundle. It is well-known that the classes~$\alpha_{g,n}(\eps)$ form a cohomological field theory. 

Let $F(\eps)$ be the potential of the deformed cohomological field theory. There is the following formula (see e.g. \cite{BPS12a}):
$$
\left.\frac{\d F(\eps))}{\d\eps}\right|_{\eps=0}=-\frac{B_{2l}}{(2l)!}\widehat{z^{2l-1}}(F),
$$
where $\widehat{z^{2l-1}}$ is the operator that acts as follows:
$$
\widehat{z^{2l-1}}(F):=-\frac{\d F}{\d t_{2l}}+\sum_{d\ge 0}t_d\frac{\d F}{\d t_{d+2l-1}}+\frac{\hbar}{2}\sum_{i+j=2l-2}(-1)^{i+1}\frac{\d^2 F}{\d t_i\d t_j}+\frac{1}{2}\sum_{i+j=2l-2}(-1)^{i+1}\frac{\d F}{\d t_i}\frac{\d F}{\d t_j}.
$$

Consider the second derivatives $\Omega_{p,q}(\eps):=\frac{\d^2 F(\eps)}{\d t_p\d t_q}$. They are differential polynomials in $u_i(\eps):=\d_x^i\frac{\d^2 F(\eps)}{\d t_0^2}$ (see \cite{BPS12a}). Denote $u_i(\eps)$ by $u_i$. Let $\frac{\d\Omega_{p,q}(\eps)}{\d\eps}[u]$ be the derivative of~$\Omega_{p,q}(\eps)$ as a differential polynomial in $u_i$. In other words,
$$
\frac{\d\Omega_{p,q}(\eps)}{\d\eps}[u]:=\frac{\d\Omega_{p,q}(\eps)}{\d\eps}-\sum_{i\ge 0}\frac{\d\Omega_{p,q}(\eps)}{\d u_i}\frac{\d u_i}{\d\eps}.
$$
In \cite{BPS12a} it is proved that
$$
\left.\frac{\d\Omega_{p,q}(\eps)}{\d\eps}[u]\right|_{\eps=0}=-\frac{B_{2l}}{(2l)!}\widehat{z^{2l-1}}[u](\Omega_{p,q}),
$$
where
\begin{align}
\widehat{z^{2l-1}}[u](\Omega_{p,q}):=&\Omega_{p+2l-1,q}+\Omega_{p,q+2l-1}+\sum_{i=0}^{2l-2}(-1)^{i+1}\Omega_{p,i}\Omega_{2l-2-i,q}\label{eq: deformation formula}\\
&-\sum_{n\ge 0}\frac{\d\Omega_{p,q}}{\d u_n}\left((n+2)\d_x^n\Omega_{0,2l-1}+\sum_{i=0}^{2l-2}\sum_{k=0}^{n-1}{n\choose k}(-1)^{i+1}\d_x^{k+1}\Omega_{0,i}\d_x^{n-k-1}\Omega_{2l-2-i,0}\right.\notag\\
&+\left.\sum_{i=0}^{2l-2}(-1)^{i+1}\d_x^n(\Omega_{0,i}\Omega_{2l-2-i,0})\right)\notag\\
&+\frac{\hbar}{2}\sum_{n,m\ge 0}\frac{\d^2\Omega_{0,2}}{\d u_n\d u_m}\sum_{i=0}^{2l-2}(-1)^{i+1}\d_x^{n+1}\Omega_{0,i}\d_x^{m+1}\Omega_{2l-2-i,0}.\notag
\end{align}

\subsection{Coefficient of $\hbar^2$}\label{subsection: coefficient of h2}

Let us return to the case of the cohomological field theory~\eqref{Hodge classes}: $\Omega_{p,q}=\frac{\d^2 F^{Hodge}}{\d t_p\d t_q}$. Let $F^{KdV}$ be the potential of the trivial cohomological field theory:
$$
F^{KdV}:=\sum_{\substack{g\ge 0,n\ge 1\\2g-2+n>0}}\frac{\hbar^g}{n!}\sum_{k_1,\ldots,k_n\ge 0}\<\tau_{k_1}\ldots\tau_{k_n}\>_g t_{k_1}\ldots t_{k_n}
$$
and $\Omega^{KdV}_{p,q}:=\frac{\d^2 F^{KdV}}{\d t_p\d t_q}$. 

From Section~\ref{subsection: deformations} it follows that 
$$
\left.\frac{\d\Omega_{0,2}}{\d\eps}[u]\right|_{\eps=0}=-\frac{1}{12}\widehat{z^1}[u]\left(\Omega^{KdV}_{0,2}\right),
$$
where
\begin{align*}
\widehat{z^1}[u]\left(\Omega^{KdV}_{0,2}\right)=&\Omega^{KdV}_{1,2}+\Omega^{KdV}_{0,3}-\Omega_{0,0}\Omega^{KdV}_{0,2}\\
&-\sum_{n\ge 0}\frac{\d\Omega^{KdV}_{0,2}}{\d u_n}\left[(n+2)\d_x^n\Omega^{KdV}_{0,1}-\sum_{k=0}^{n-1}{n\choose k}u_{k+1}u_{n-k-1}-\d_x^n(u^2)\right]\\
&-\frac{\hbar}{2}\sum_{n,m\ge 0}\frac{\d^2\Omega^{KdV}_{0,2}}{\d u_m\d u_n}u_{n+1}u_{m+1}.
\end{align*}

We have the following formulas (see e.g. \cite{DZ05}):
\begin{align*}
&\Omega^{KdV}_{0,1}=\frac{u^2}{2}+\frac{\hbar}{12}u_2,\\
&\Omega^{KdV}_{0,2}=\frac{u^3}{6}+\frac{\hbar}{24}(u_1^2+2u u_2)+\frac{\hbar^2}{240}u_4,\\
&\Omega^{KdV}_{0,3}=\frac{u^4}{24}+\frac{\hbar}{24}(u^2u_2+uu_1^2)+\frac{\hbar^2}{480}(2uu_4+4u_1u_3+3u_2^2)+\frac{\hbar^3}{6720}u_6,\\
&\Omega^{KdV}_{1,2}=\frac{u^4}{8}+\frac{\hbar}{24}(3u^2u_2+2uu_1^2)+\hbar^2\left(\frac{uu_4}{90}+\frac{23u_2^2}{1440}+\frac{u_1u_3}{60}\right)+\frac{\hbar^3}{2880}u_6.
\end{align*}
By direct computations, we get
$$
\int\widehat{z^1}[u]\left(\Omega_{0,2}^{KdV}\right)dx=\int\left(\frac{\hbar}{4}u^2u_2+\frac{\hbar^2}{30}u u_4\right)dx.
$$

Thus, the coefficient of $\hbar^2\eps$ in $\int\Omega_{0,2}dx$ is equal to $-\frac{1}{360}\int u u_4dx$. Now it is easy to compute that the coefficient $c_2$ in~\eqref{eq: unknown coefficients} is equal to $\frac{1}{1440}$. This completes the proof of Proposition~\ref{proposition: pr1}.


\section{Proof of Lemma~\ref{lemma: identity}}\label{section: proof of the identity}

Introduce the function $\phi(z):=\sum_{i\ge 0}\frac{B_{2i}}{(2i)!}z^{2i}$. For a power series $f(z)=\sum_{i\ge 0}f_iz^i$, we denote by $[z^i]f$ the coefficient $f_i$. The coefficient of $u_{2k+1}u_{2g-2k}$ on the left-hand side of \eqref{eq: Bernoulli identity} is equal to
\begin{align*}
&\frac{B_{2k}B_{2g-2k}}{(2k)!(2g-2k)!}-\sum_{i=0}^g{2i\choose 2k}\frac{B_{2i}B_{2g-2i}}{(2i)!(2g-2i)!}-\sum_{i=0}^k{2g-2i\choose 2g-2k-1}\frac{B_{2i}B_{2g-2i}}{(2i)!(2g-2i)!}=\\
=&[z^{2g}]\left(\frac{B_{2k}\phi z^{2k}}{(2k)!}-\frac{\phi\phi^{(2k)}z^{2k}}{(2k)!}-\sum_{i=0}^{k}\frac{B_{2i}\phi^{(2k-2i+1)}z^{2k+1}}{(2i)!(2k-2i+1)!}\right).
\end{align*}
Therefore, the lemma is equivalent to the following identity.
\begin{gather*}
\frac{B_{2k}\phi z^{2k}}{(2k)!}-\frac{\phi\phi^{(2k)}z^{2k}}{(2k)!}-\sum_{i=0}^{k}\frac{B_{2i}\phi^{(2k-2i+1)}z^{2k+1}}{(2i)!(2k-2i+1)!}=-\delta_{k,0}\frac{z^2}{4}.
\end{gather*}

Let us rewrite it in a bit different way:
\begin{gather}
\label{eq: first}
\frac{\phi\phi^{(2k)}}{(2k)!}=\frac{B_{2k}\phi}{(2k)!}-\sum_{i=0}^{k}\frac{B_{2i}\phi^{(2k-2i+1)}z}{(2i)!(2k-2i+1)!}+\delta_{k,0}\frac{z^2}{4}. 
\end{gather}
Let us formulate another identity of this type. 
\begin{gather}
\label{eq: second}
\frac{\phi^{(2k+1)}\phi}{(2k+1)!}=-\sum_{i=0}^k\frac{B_{2i}\phi^{(2k+2-2i)}z}{(2i)!(2k+2-2i)!}+\delta_{k,0}\frac{z}{4}.
\end{gather}

We prove \eqref{eq: first} and \eqref{eq: second} by induction on $k$. For $k=0$, equation~\eqref{eq: first} looks as follows:
\begin{gather}\label{eq: first derivative}
z\phi'=-\phi^2+\phi+\frac{z^2}{4}.
\end{gather}
It is equivalent to the following identity between the Bernoulli numbers (see e.g. \cite{GKP94}).
$$
\sum_{i=1}^m\frac{B_{2i}B_{2m-2i}}{(2i)!(2m-2i)!}=-\frac{2m B_{2m}}{(2m)!}+\frac{\delta_{m,1}}{4}.
$$

Suppose that \eqref{eq: first} is true and also \eqref{eq: second} is true for $k'<k$. Let us prove \eqref{eq: second}. Let us differentiate~\eqref{eq: first}, we get
$$
\frac{\phi'\phi^{(2k)}}{(2k)!}+\frac{\phi\phi^{(2k+1)}}{(2k)!}=\frac{B_{2k}\phi'}{(2k)!}-\sum_{i=0}^{k}\frac{B_{2i}\phi^{(2k-2i+2)}z}{(2i)!(2k-2i+1)!}-\sum_{i=0}^{k}\frac{B_{2i}\phi^{(2k-2i+1)}}{(2i)!(2k-2i+1)!}+\delta_{k,0}\frac{z}{2}.
$$ 
Using~\eqref{eq: first derivative} and the induction assumption, we get
\begin{gather}\label{eq: tmp1}
\left(-\frac{\phi^2}{z}+\frac{z}{4}\right)\frac{\phi^{(2k)}}{(2k)!}+\frac{\phi\phi^{(2k+1)}}{(2k)!}=\left(-\frac{\phi^2}{z}+\frac{z}{4}\right)\frac{B_{2k}}{(2k)!}-\sum_{i=0}^{k}\frac{B_{2i}\phi^{(2k-2i+2)}z}{(2i)!(2k-2i+1)!}+\delta_{k,0}\frac{z}{4}.
\end{gather}
From the induction assumption it follows that
\begin{gather}\label{eq: tmp2}
\frac{\phi^2\phi^{(2k)}}{(2k)!}=\frac{B_{2k}\phi^2}{(2k)!}-\frac{\phi\phi^{(2k+1)}z}{(2k+1)!}-\sum_{i=0}^k\frac{B_{2i}\phi^{(2k-2i+2)}z^2}{(2i-1)!(2k-2i+2)!}+\frac{\phi^{(2k)}z^2}{4(2k)!}-\frac{B_{2k}z^2}{4(2k)!}+\frac{\delta_{k,0}z^2}{4}.
\end{gather}
After substituting \eqref{eq: tmp2} into \eqref{eq: tmp1} we get \eqref{eq: second}.

Suppose that \eqref{eq: second} is true and also \eqref{eq: first} is true, for any $k'\le k$. Then the proof of \eqref{eq: first} for $k'=k+1$ can be done in a completely similar way. This concludes the proof of the lemma.

\end{document}